\documentclass[11pt]{article}
\usepackage[left=1in,top=1in,right=1in,bottom=1in,head=.1in,nofoot]{geometry}

\usepackage{setspace,url,bm,amsmath} 
\usepackage{float}
\usepackage{caption}
\usepackage{subcaption,siunitx,booktabs}

\usepackage{titlesec} 
\usepackage{rotating}
\titlelabel{\thetitle.\quad} 
\titleformat*{\section}{\bf\large}

\usepackage{graphicx} 
\usepackage{bbm}
\usepackage{latexsym}
\usepackage{color}
\usepackage{authblk}
\usepackage{amsthm}
\usepackage{amsfonts}

\theoremstyle{definition}
\newtheorem{theorem}{Theorem}
\newtheorem{proposition}{Proposition}
\newtheorem{lemma}{Lemma}
\newtheorem{example}{Example}

\newtheorem{remark}{Remark}

\usepackage{natbib} 
\bibpunct{(}{)}{;}{a}{}{,} 
\usepackage{enumitem}

\usepackage{etoolbox} 
\apptocmd{\sloppy}{\hbadness 10000\relax}{}{} 

\def\var{\textnormal{var}}
\def\cov{\textnormal{cov}}
\def\rank{\textnormal{rank}}

\def\obs{\textnormal{obs}}
\def\bt{\textnormal{bt}}
\def\win{\textnormal{in}}

\def\digamma{\text{digamma}}

\def\-{\mbox{-}}

\title{\bf Causal Inference from Possibly Unbalanced Split-Plot Designs: A Randomization-based Perspective}

\author[*]{Rahul Mukerjee}
\author[**]{Tirthankar Dasgupta}

\affil[*]{Indian Institute of Management Calcutta, Joka, Diamond Harbour Road, Kolkata 700104, India, email: rmuk0902@gmail.com}
\affil[***]{Department of Statistics, Rutgers University, 110 Frelinghuysen Road, Piscataway, New Jersey 08901, U.S.A., email: tirthankar.dasgupta@rutgers.edu}

\begin{document}

\date{}

\doublespacing

\maketitle

\begin{abstract}
\frenchspacing
Split-plot designs find wide applicability in multifactor experiments with randomization restrictions. Practical considerations often warrant the use of unbalanced designs. This paper investigates randomization based causal inference in split-plot designs that are possibly unbalanced. Extension of ideas from the recently studied balanced case yields an expression for the sampling variance of a treatment contrast estimator as well as a conservative estimator of the sampling variance. However, the bias of this variance estimator does not vanish even when the treatment effects are strictly additive. A careful and involved matrix analysis is employed to overcome this difficulty, resulting in a new variance estimator, which becomes unbiased under milder conditions. A construction procedure that generates such an estimator with minimax bias is proposed. 
\end{abstract}

\textbf{Keywords:} 
Bias; Factorial experiment; Finite population; Minimaxity; Treatment-effect additivity.

\section{Introduction}
\label{sec:intro}
\frenchspacing

Factorial experiments were originally developed in the context of agricultural experiments \citep{Fisher1925,Fisher1935,Yates::1935} and later extensively used in industrial and engineering applications \citep{WandH::2009}. Such experiments have currently been undergoing a third popularity surge among social, behavioral, and biomedical sciences. However, one of the key challenges of using standard principles of designing and analyzing factorial experiments in these fields arises from randomization restrictions. Consider a simplified version of the education experiment described in \cite{DPR::2015}. Suppose the goal is to assess the causal effects of two interventions (referred to as factors in experimental design literature) -- $F_1$: a mid-year quality review by a team of experts, and $F_2$: a bonus scheme for teachers -- on the performances of 40 schools in the state of New York. Each factor has two levels denoted by 1 (application) and 0 (non-application). A completely randomized assignment of the 40 schools to the four treatment combinations $00, 01, 10, 11$ is likely to disperse the schools assigned to level 1 of factor $F_1$ (i.e., schools to undergo review) all over the state. Such a design may be prohibitive from the consideration of travel cost and time. A more practical alternative would be to divide these 40 schools by geographic proximity into four groups called whole-plots. Two of these whole-plots would then be randomly assigned to level 0 and the other two to level 1 of factor $F_1$. The teacher bonus scheme can then be applied to half of the schools chosen randomly within each whole-plot. Such a randomization scheme is an example of a classic split-plot design. See \cite{Kirk::1982}, \cite{CC::1957}, \cite{BHH::2005}, and \cite{WandH::2009} for formal definitions.

Randomization-based inference is the most natural methodology to draw inference on causal effects of treatments from split-plot experiments in a finite population setting, as observed by \cite{Freedman2006,Freedman2008a}. Recently, \cite{zhaosplitplot::2018} developed a framework for randomization-based estimation procedure of finite-population causal effects for balanced split-plot designs, in which each whole-plot consists of the same number of units or sub-plots, and any treatment combination of the sub-plot factors occurs equally often in all whole-plots; vide (\ref{eq:balance_condition}) below. However, unbalanced split-plot designs are quite common in the social sciences. Consider the school experiment described earlier. Suppose the 40 schools are spread over four counties with 8, 8, 12 and 12 schools in these counties. In this case, each county can be considered as a natural whole-plot. Thus the design is unbalanced and the estimation methodology proposed by \cite{zhaosplitplot::2018} is no longer applicable.

In this paper we investigate randomization based causal inference in split-plot designs that are possibly unbalanced, using the potential outcomes framework \citep{neyman::1923,Rubin1974,Rubin1978,Rubin2005}. We start with a natural unbiased estimator of a typical treatment contrast and first examine how far the approach of \cite{zhaosplitplot::2018} for the balanced case can be adapted to our more general setup. It is seen that this approach, aided by a variable transformation, yields an expression for the sampling variance of the treatment contrast estimator but runs into difficulty in variance estimation. Specifically, as in the balanced case and other situations in causal inference, the resulting variance estimator is conservative in the sense of having a nonnegative bias. However, unlike in most standard situations, the bias does not vanish even under strict additivity or homegeneity of treatment effects. To overcome this problem, a careful matrix analysis is employed leading, under wide generality, to a new variance estimator. This estimator is also conservative, but enjoys the nice property of becoming unbiased under between-whole-plot additivity, a condition even milder than strict additivity. We also discuss the issue of minimaxity, with a view to controlling the bias in variance estimation and explore the bias of the estimator under treatment effect heterogeneity via simulations. 

\section{Treatment contrast and its unbiased estimation}
\label{sec:contrast_estimation}

Consider a factorial experiment conducted to assess causal effects of $m_1$ whole-plot factors $F_{11}, \ldots, F_{1 m_1}$ and $m_2$ sub-plot factors $F_{21}, \ldots, F_{2 m_2}$ on a finite population of $N$ units. Each factor has two or more levels. The treatment combinations are denoted by $z = z_1 z_2$, where $z_k \in  Z_k$ and $Z_k$ is the set of level combinations of $F_{k1}, \ldots, F_{k m_k} \ (k =1,2)$. For $i = 1, \ldots, N$, let $Y_i(z_1 z_2)$ denote the potential outcome of unit $i$ when exposed to treatment combination $z_1z_2$. 
A typical treatment contrast for unit $i$ of the form
\begin{equation}
\tau_i = \sum_{z_1\in Z_1} \sum_{z_2\in Z_2} g(z_1z_2) Y_i(z_1z_2), \label{eq:unitlevel}
\end{equation} 
where $g(z_1 z_2)$, $z_1 \in Z_1, z_2 \in Z_2$ are known, not all zeros, and sum to zero. Let 
\begin{equation}
\overline{Y}(z_1z_2) = N^{-1} \sum_{i=1}^N Y_i(z_1 z_2), \label{eq:popmeanz1z2}
\end{equation}
denote the average potential outcome for treatment combination $z_1 z_2$, and let 
\begin{equation}
\overline{\tau} = N^{-1} \sum_{i=1}^N \tau_i = \sum_{z_1\in Z_1} \sum_{z_2\in Z_2} g(z_1z_2) \overline{Y}(z_1z_2), \label{eq:poplevel}
\end{equation}
denote a treatment contrast for the finite population of $N$ units. We define $\overline{\tau}$ as the finite-population causal estimand of interest and consider the problem of drawing inference on $\overline{\tau}$ using the outcomes observed from the experiment.

The observed outcomes are generated through an assignment mechanism, which is the process of allocating treatment combinations to the $N$ units. Here we consider a split-plot assignment mechanism which can be described as follows. Suppose there is a partitioning of the $N$ experimental units into $W (\ge 2)$ disjoint sets $\Omega_1, \ldots, \Omega_W$, called whole-plots, such that $\Omega_w$ consists of $M_w (\ge 2)$ units, called sub-plots, $w = 1, \ldots, W$, and $M_1 + \ldots M_W = N$. Consider now a two-stage randomization, which assigns $r_1(z_1)$ whole-plots to level combination $z_1$ of $F_{11}, \ldots F_{1 m_1}$ and then, for each $w = 1, \ldots, W$, assigns $r_{w2}(z_2)$ sub-plots within whole-plot $\Omega_w$ to level combination $z_2$ of $F_{21}, \ldots F_{2 m_2}$. Here at each stage all assignments are equiprobable, the $r_1(z_1)$ and $r_{w2}(z_2)$ are fixed positive integers, and $\sum_{z_1 \in Z_1} r_1(z_1) = W$, $\sum_{z_2 \in Z_2} r_{w2}(z_2) = M_w$ for $w = 1, \ldots, W$.
    
Note that the above assignment mechanism yields a \textit{balanced} split-plot design if
\begin{equation}
M_1 = \cdots = M_W, \quad r_{12}(z_2) = \cdots = r_{W2}(z_2), \ \mbox{for all} \ z_2 \in Z_2. \label{eq:balance_condition}
\end{equation}
In the school example described in Section \ref{sec:intro}, the whole-plots represent sets of schools within a county and we have $N = 40$, $W = 4$, $M_1 = M_2 =8$, $M_3 = M_4 = 12$, $Z_1 = Z_2 = \{0,1\}$. Finally, for all $z_2 \in Z_2$, $r_{w2}(z_2) = 4$ for $w = 1,2$ and $r_{w2}(z_2) = 6$ for $w = 3,4$. Thus, the design is unbalanced.

To define the observed outcomes of the experiment, we introduce two sets of random treatment assignment indices at the whole-plot and the sub-plot levels. Let $T_1(z_1)$ denote the set of indices $w$ such that whole-plot $\Omega_w$ is randomly assigned to level combination $z_1$ of $F_{11}, \ldots, F_{1 m_1}$. Similarly, for $z_2 \in Z_2$ and $w = 1, \ldots, W$, let $T_{w2}(z_2)$ be the set of sub-plots in $\Omega_w$ randomly assigned to level combination $z_2$ of $F_{21}, \ldots, F_{2 m_2}$. 
For any treatment combination $z_1 z_2$, the observed outcomes from the whole-plot $\Omega_w$, $w \in T_1(z_1)$, are then $Y_i(z_1 z_2)$, $i \in T_{w2}(z_2)$. Let
\begin{equation}
\overline{Y}_w^{\obs}(z_1 z_2) = \left\{ r_{w2}(z_2) \right\}^{-1} \sum_{i \in T_{w2}(z_2)} Y_i(z_1 z_2), \label{eq:WPmean_estimator} 
\end{equation}
denote the average observed outcome for treatment combination $z_1 z_2$ within whole-plot $\Omega_w$ for $w \in T_1(z_1)$. In the spirit of the usual unbiased estimator of the population mean in two-stage sampling \citep{cochran::1977}, define
\begin{equation}
\overline{Y}^{\obs} (z_1 z_2) = \frac {W}{N r_1(z_1)} \sum_{w \in T_1(z_1)} M_w \overline{Y}_w^{\obs}(z_1 z_2) =  \frac{1} {r_1(z_1)} \sum_{w \in T_1(z_1)} \frac{M_w}{\overline{M}} \overline{Y}_w^{\obs}(z_1 z_2), \label{eq:mean_estimator}
\end{equation}
where $\overline{M} = (M_1 + \ldots + M_W)/W = N/W$ is the average whole-plot size. From (\ref{eq:WPmean_estimator}) and (\ref{eq:mean_estimator}), it is straightforward to verify by conditioning on the randomization at the whole-plot level that $E \left\{ \overline{Y}^{\obs}(z_1 z_2) \right\} = \overline{Y}(z_1 z_2)$, where $\overline{Y}(z_1 z_2)$ is given by (\ref{eq:popmeanz1z2}). Using (\ref{eq:poplevel}), an immediate consequence of this fact is Proposition \ref{prop:unbiased}.

\begin{proposition} \label{prop:unbiased}
An unbiased estimator of the finite population treatment contrast $\overline{\tau}$ is given by
\begin{equation}
\widehat{\overline{\tau}} = \sum_{z_1\in Z_1} \sum_{z_2\in Z_2} g(z_1z_2) \overline{Y}^{\obs} (z_1 z_2), \label{eq:taubarest} 
\end{equation}
where $\overline{Y}^{\obs} (z_1 z_2)$ is given by (\ref{eq:mean_estimator}).
\end{proposition}

\section{Sampling variance and its estimation generalizing the balanced case}
\label{sec:var_taubarhat}

Proposition \ref{prop:unbiased} yields a point estimator of $\overline{\tau}$. However, to quantify the uncertainty associated with the point estimator and draw inference on $\overline{\tau}$, one needs to derive and estimate the sampling variance of $\widehat{\overline{\tau}}$ with respect to its distribution induced by the randomization in the split-plot design. \citet{zhaosplitplot::2018} derived an expression for the sampling variance of $\widehat{\overline{\tau}}$ for a balanced split-plot design, that is, when conditions (\ref{eq:balance_condition}) are satisfied. They also obtained an estimator of the sampling variance that, like most variance estimators in finite population causal inference \citep{MDR::2018}, has a nonnegative bias. Further, they noted that this bias vanishes under between-whole-plot additivity, that is, average treatment effect homogeneity at the whole-plot level. In this Section, we derive an expression for the sampling variance and find a variance estimator generalizing the arguments in \cite{zhaosplitplot::2018} to the unbalanced case, and examine the properties of the estimator. To that end, we first convert the ``raw'' potential outcomes $Y_i(z_1 z_2)$ to ``adjusted'' potential outcomes
\begin{equation}
U_i(z_1 z_2) = (M_w / \overline{M}) Y_i(z_1 z_2), \label{eq:adjusted_PO}
\end{equation}
for each $z_1 \in Z_1$, $z_2 \in Z_2$, $i \in \Omega_w$ and $w = 1, \ldots, W$. An intuition behind this adjustment will be provided shortly, after we introduce its observed version.
 
For each $z_1 z_2$, define $\overline{U}_w(z_1z_2) = M_w^{-1} \sum_{i \in \Omega_w} U_i(z_1 z_2), \ w=1, \ldots, W$, and $\overline{U}(z_1 z_2) = W^{-1} \sum_{w=1}^W \overline{U}_w (z_1 z_2)$. By (\ref{eq:adjusted_PO}), $\overline{U}(z_1 z_2) = \overline{Y}(z_1 z_2)$. Next, for $z_1, z_1^* \in Z_1$ and $z_2, z_2^* \in Z_2$, define
\begin{eqnarray*}
S_{\bt}(z_1 z_2, z_1^* z_2^*) = \frac{\overline{M}}{W-1} \sum_{w=1}^W \left\{ \overline{U}_w(z_1z_2) - \overline{U}(z_1 z_2) \right\}  \left\{ \overline{U}_w(z_1^* z_2^*) - \overline{U}(z_1^* z_2^*) \right\}, \label{eq:Sbet} \\
S_{\win, w}(z_1 z_2, z_1^* z_2^*) = \frac{1}{M_w-1}  \sum_{i \in \Omega_w} \left\{ U_i(z_1 z_2) - \overline{U}_w(z_1z_2) \right\} \left\{ U_i(z_1^* z_2^*) - \overline{U}_w(z_1^* z_2^*) \right\}. \label{eq:Sin} 
\end{eqnarray*}
In the balanced case, $S_{\bt}(z_1 z_2, z_1^* z_2^*)$ and $W^{-1} \sum_{w=1}^W S_{\win, w}(z_1 z_2, z_1^* z_2^*)$ represent, respectively, the between and within whole-plot mean squares or products in an analysis of variance/covariance decomposition of the potential outcomes.

It is also important to define a measure of heterogeneity of treatment contrasts across the whole-plots. First, Let
\begin{equation}
\overline{\tau}_w = (1/M_w) \sum_{i \in \Omega_w} \tau_i = \sum_{z_1 \in Z_1} \sum_{z_2 \in Z_2} g(z_1z_2) \overline{Y}_w(z_1 z_2), \ w=1, \ldots, W, \label{eq:WPtau}
\end{equation} 
denote the whole-plot level treatment contrasts, where 
$ \overline{Y}_w(z_1 z_2) =  M_w^{-1} \sum_{i \in \Omega_w} Y_i(z_1 z_2) $
is the average potential outcome of all units in whole-plot $\Omega_w$ for treatment combination $z_1 z_2$. The second equality in (\ref{eq:WPtau}) follows from (\ref{eq:unitlevel}). Also, from (\ref{eq:poplevel}) and (\ref{eq:WPtau}), it follows that
\begin{equation}
\overline{\tau} = (1/W) \sum_{w=1}^W (M_w/\overline{M}) \overline{\tau}_w. \label{eq:connect_tau}
\end{equation}
Now define the following measure of heterogeneity of treatment contrasts across the whole-plots:
\begin{equation}
\Delta = \frac{1}{W(W-1)} \sum_{w=1}^W \left\{ (M_w/\overline{M}) \overline{\tau}_w - \overline{\tau} \right\}^2, \label{eq:delta}
\end{equation}
where $\overline{\tau}_w$ is given by (\ref{eq:WPtau}). Then, extending the ideas of \citet{zhaosplitplot::2018}, after considerable algebra, we obtain the following result on the sampling variance of $\widehat{\overline{\tau}}$, the unbiased estimator of $\overline{\tau}$.

\begin{theorem} \label{thm:tauhat_variance}
The sampling variance of $\widehat{\overline{\tau}}$ is
\begin{eqnarray*}
\var(\widehat{\overline{\tau}}) &=& \sum_{z_1 \in Z_1} \sum_{z_2 \in Z_2} \sum_{z_2^* \in Z_2} \frac{g(z_1 z_2) g(z_1 z_2^*)}{r_1(z_1)} \left( \frac{S_{\bt}(z_1 z_2, z_1 z_2^*)}{\overline{M}} - \sum_{w=1}^W \frac{S_{\win, w}(z_1 z_2, z_1 z_2^*)}{W M_w} \right) \\
&+& \sum_{z_1 \in Z_1} \sum_{z_2 \in Z_2} \frac{ \left\{g(z_1 z_2) \right\}^2}{W r_1(z_1)} \sum_{w=1}^W \frac{S_{\win, w}(z_1 z_2, z_1 z_2)}{r_{w2}(z_2)} - \Delta.
\end{eqnarray*}
\end{theorem}

Next, to obtain an estimator of the sampling variance, we first define the counterparts of $\overline{Y}_w^{\obs}(z_1 z_2)$ and $\overline{Y}^{\obs}(z_1 z_2)$ in (\ref{eq:WPmean_estimator}) and (\ref{eq:mean_estimator}) in terms of the adjusted potential outcomes:
\begin{eqnarray*}
\overline{U}_w^{\obs}(z_1 z_2) &=& \frac{1}{r_{w2}(z_2)} \sum_{i \in T_{w2}(z_2)} U_i(z_1 z_2), \quad w \in T_1(z_1) \quad \mbox{and} \\ \overline{U}^{\obs}(z_1 z_2) &=& \frac{1}{r_1(z_1)} \sum_{w \in T_1(z_1)}\overline{U}_w^{\obs}(z_1 z_2).
\end{eqnarray*}
Then it is easy to see from (\ref{eq:WPmean_estimator}), (\ref{eq:mean_estimator}) and (\ref{eq:adjusted_PO}) that
\begin{equation}
\overline{Y}^{\obs}(z_1 z_2) = \overline{U}^{\obs}(z_1 z_2). \label{eq:Y-u-equality}
\end{equation}
Note that $\overline{U}^{\obs}(z_1 z_2)$ is the simple average of $\overline{U}_w^{\obs}(z_1 z_2)$, $w \in T_1(z_1)$, irrespective of whether $M_1, \ldots, M_W$ are equal or not. This is precisely what the relationship between $\overline{Y}^{\obs}(z_1 z_2)$ and $\overline{Y}_w^{\obs}(z_1 z_2)$ in (\ref{eq:mean_estimator}) reduces to when $M_1 = \cdots = M_W$, providing us with the intuition to generalize the results of \citet{zhaosplitplot::2018} by substituting the potential outcomes by their adjusted version in view of (\ref{eq:Y-u-equality}). We now define the following estimator of the sampling variance in Theorem \ref{thm:tauhat_variance}:
\begin{equation}
\widehat{V}(\widehat{\overline{\tau}}) = \sum_{z_1 \in Z_1} \sum_{z_2 \in Z_2} \sum_{z_2^* \in Z_2} \frac{g(z_1 z_2) g(z_1 z_2^*)}{r_1(z_1)} \widehat{S}(z_1 z_2, z_1 z_2^*), \label{eq:var_est1}
\end{equation}
where
\begin{eqnarray*}
\widehat{S}(z_1 z_2, z_1 z_2^*) = \frac{1}{r_1(z_1) - 1} \sum_{w \in T_1(z_1)} \left\{ \overline{U}_w^{\obs}(z_1 z_2) - \overline{U}^{\obs}(z_1 z_2)  \right\} \left\{ \overline{U}_w^{\obs}(z_1 z_2^*) - \overline{U}^{\obs}(z_1 z_2^*) \right\}. \label{eq:S_hat}
\end{eqnarray*}

These expressions now allow us to work along the lines of \citet{zhaosplitplot::2018} by substituting (\ref{eq:Y-u-equality}) in  (\ref{eq:taubarest}). Again, considerable algebra yields the following result:

\begin{theorem} \label{thm:bias1}
The variance estimator $\widehat{V}(\widehat{\overline{\tau}})$ given by (\ref{eq:var_est1}) estimates the sampling variance of $\widehat{\overline{\tau}}$ with a nonnegative bias $\Delta$ defined by (\ref{eq:delta}), that is, $ E \left\{ \widehat{V}(\widehat{\overline{\tau}}) \right\} = \var(\widehat{\overline{\tau}}) + \Delta$.
\end{theorem}

\begin{remark} \label{remark::theorem2}
Theorem \ref{thm:bias1} shows that  $\widehat{V}(\widehat{\overline{\tau}})$ is a conservative estimator of $\var(\widehat{\overline{\tau}})$ with a non-negative bias $\Delta$. This property is in line with variance estimators in other situations of randomization based causal inference. Moreover, in the balanced case, by (\ref{eq:delta}), the bias $\Delta$ vanishes when $\overline{\tau}_1 = \cdots =\overline{\tau}_W = \overline{\tau}$. As observed by \citet{zhaosplitplot::2018}, this happens for every treatment contrast $\overline{\tau}$ if and only if between-whole-plot additivity holds, which means
\begin{equation}
\overline{Y}_1(z_1 z_2) - \overline{Y}_1(z_1^* z_2^*) = \cdots = \overline{Y}_W(z_1 z_2) - \overline{Y}_W(z_1^* z_2^*), \label{eq:WP-additivity}  
\end{equation}   
for every pair of treatment combinations $z_1 z_2$ and $z_1^* z_2^*$. A disturbing feature of the variance estimator $\widehat{V}(\widehat{\overline{\tau}})$, however, emerges in the unbalanced case which is the main focus of this paper. Then  $\widehat{V}(\widehat{\overline{\tau}})$ remains biased even if between-whole-plot additivity holds, because by (\ref{eq:WPtau}) and (\ref{eq:connect_tau}), condition (\ref{eq:WP-additivity}) implies $\overline{\tau}_1 = \cdots =\overline{\tau}_W = \overline{\tau}$ and hence
\begin{eqnarray*}
\Delta = \frac{\overline{\tau}^2}{W(W-1) \overline{M}^2} \sum_{w=1}^W (M_w - \overline{M})^2, 
\end{eqnarray*}
which is positive when $M_1, \ldots, M_W$ are not all equal unless $\overline{\tau} = 0$. The situation remains unchanged even under the stronger assumption of strict additivity or homogeneity of treatment effects \citep{neyman::1923}, which enforces the constancy of $Y_i(z_1 z_2) - Y_i(z_1^* z_2^*)$ over $i=1, \ldots, N$ for every pair of treatment combinations $z_1 z_2$ and $z_1^* z_2^*$. 
\end{remark}

This property of $\widehat{V}(\widehat{\overline{\tau}})$ described in Remark \ref{remark::theorem2} is a matter of concern because a requirement typically imposed on a variance estimator in causal inference is that it should become unbiased at least under Neymannian strict additivity, if not under milder versions thereof such as between-whole-plot additivity in the present context. The estimator $\widehat{V}(\widehat{\overline{\tau}})$, obtained by generalizing the arguments in the balanced case fails to meet this requirement when $M_1, \ldots, M_W$ are not all equal. In the rest of the paper, we investigate the existence of a variance estimator that overcomes this difficulty and show how, under wide generality, such an estimator can be obtained by appropriately modifying $\widehat{V}(\widehat{\overline{\tau}})$ as given by (\ref{eq:var_est1}).

\section{A new variance estimator}
\label{sec::new_estimator}

We begin our search for an improved variance estimator by expanding the bias term $\Delta$ defined in (\ref{eq:delta}) as follows:
\begin{equation}
\Delta = (1/N)^2 \left[ \sum_{w=1}^W M_w^2 \overline{\tau}_w^2 - \sum_{w=1}^W \sum_{w^* (\ne w) =1}^W \left \{ M_w M_{w^*} / (W-1) \right \} \overline{\tau}_w \overline{\tau}_{w^*}  \right]. \label{eq:delta_mod}
\end{equation}
Note that in (\ref{eq:delta_mod}), the term $\overline{\tau}_w^2$ is not unbiasedly estimable, but for $w \ne w^*$, $\overline{\tau}_w \overline{\tau}_{w^*}$ allows unbiased estimation. This is because, by (\ref{eq:WPtau}),
\begin{equation}
\overline{\tau}_w^2 = (1/M_w)^2 \sum_{z_1 \in Z_1} \sum_{z_2 \in Z_2} \sum_{z_1^* \in Z_1} \sum_{z_2^* \in Z_2} \sum_{i \in \Omega_w} \sum_{i^* \in \Omega_w} g(z_1 z_2) g(z_1^* z_2^*) Y_i(z_1 z_2) Y_{i^*}(z_1^* z_2^*). \label{eq:tauw_sq} 
\end{equation}
The sums over $i$ and $i^*$ in (\ref{eq:tauw_sq}) include the case $i = i^*$. There is at least one pair of distinct treatment combinations $z_1 z_2$ and $z_1^* z_2^*$ such that $g(z_1 z_2) g(z_1^* z_2^*) \ne 0$ and $Y_i(z_1 z_2) Y_i(z_1^* z_2^*)$ is never observable as unit $i$ cannot be assigned simultaneously to both $z_1 z_2$ and $z_1^* z_2^*$. Hence, $\overline{\tau}_w^2$ does not allow unbiased estimation. On the other hand, for $w \ne w^*$, $\overline{\tau}_w \overline{\tau}_{w^*}$ does not involve terms like $Y_i(z_1 z_2) Y_i(z_1^* z_2^*)$, and is unbiasedly estimable. For each $w$, let $z_{1w}$ denote the level combination of the whole-plot factors assigned to whole-plot $\Omega_w$. Now define
\begin{eqnarray*}
G_w^{\obs} = \sum_{z_2 \in Z_2} g(z_{1w} z_2) \overline{Y}_w^{\obs} (z_{1w} z_2). \label{eq:Gw}
\end{eqnarray*}
The following proposition now gives an unbiased estimator of $\overline{\tau}_w \overline{\tau}_{w^*}$:
\begin{proposition} \label{prop::2}
For $w , w^* = 1, \ldots, W$, $w \ne w^*$, an unbiased estimator of $\overline{\tau}_w \overline{\tau}_{w^*}$ is given by
\begin{eqnarray*}
H_{w w^*} = \frac{W (W-1) G_w^{\obs} G_{w^*}^{\obs}}{r_1(z_{1w}) \left \{ r_1(z_{1w^*}) - \delta(z_{1w}, z_{1w^*}) \right\} },
\end{eqnarray*}
where $\delta(z_{1w}, z_{1w^*})$ is an indicator that equals one if $z_{1w} = z_{1w^*}$ and zero otherwise.
\end{proposition}

We can now use Proposition \ref{prop::2} to construct a new estimator of $\var(\widehat{\overline{\tau}})$. Consider any symmetric matrix $B = (( b_{w w^*} ))$ of order $W$ such that $b_{ww} = M_w^2$ for $w = 1, \ldots, W$. Now define the variance estimator
\begin{equation}
\widetilde{V}(\widehat{\overline{\tau}}) = \widehat{V}(\widehat{\overline{\tau}}) + (1/N^2)  \sum_{w=1}^W \sum_{w^* (\ne w) =1}^W \left [ b_{w w^*} + \left\{ M_w M_{w^*} / (W-1) \right \} \right] H_{w w^*}, \label{eq:new_estimator} 
\end{equation}
where $\widehat{V}(\widehat{\overline{\tau}})$ is the variance estimator defined in Section \ref{sec:var_taubarhat}, and $H_{w w^*}$ is as defined in Proposition \ref{prop::2}. Then, from (\ref{eq:delta_mod}), (\ref{eq:new_estimator}), Theorem \ref{thm:bias1} and Proposition \ref{prop::2} it is easy to see that
\begin{eqnarray*}
E \left\{ \widetilde{V}(\widehat{\overline{\tau}}) \right\}  &=& \var(\widehat{\overline{\tau}}) + \Delta + (1/N^2)  \sum_{w=1}^W \sum_{w^* (\ne w) =1}^W \left [ b_{w w^*} + \left\{ M_w M_{w^*} / (W-1) \right \} \right]\overline{\tau}_w \overline{\tau}_{w^*} \\
&=& \var(\widehat{\overline{\tau}})  + \widetilde{\Delta},
\end{eqnarray*}
where
\begin{equation}
\widetilde{\Delta} = (1/N^2) \sum_{w=1}^W \sum_{w^*=1}^W b_{w w^*}\overline{\tau}_w \overline{\tau}_{w^*}. \label{eq:newdelta}
\end{equation}

Clearly, the bias $\widetilde{\Delta}$ is nonnegative, making $\widetilde{V}(\widehat{\overline{\tau}})$ a conservative estimator of $\var(\widehat{\overline{\tau}})$ if the matrix $B$ is nonnegative definite. Furthermore, by (\ref{eq:newdelta}), this bias vanishes if and only if $\overline{\tau}_1 = \cdots = \overline{\tau}_W$, when $B$ has each row sum zero, and is a positive semidefinite matrix of rank $W-1$. These facts are summarized in Theorem \ref{thm:zerobias}, which is the main result of this section.

\begin{theorem} \label{thm:zerobias}
Let there exist a positive semidefinite matrix $B = ((b_{w w^*})) $ of order $W$ and satisfying the conditions:
(c1) $b_{ww} = M_w^2, \ w=1, \ldots, W$, (c2) $\sum_{w^*=1}^W b_{w w^*} = 0, \ w=1, \ldots, W$, and (c3) rank$(B) = W-1$. Then the variance estimator $\widetilde{V}(\widehat{\overline{\tau}})$ defined in (\ref{eq:new_estimator}) estimates  $\var(\widehat{\overline{\tau}})$ with a nonnegative bias $\widetilde{\Delta}$ given by (\ref{eq:newdelta}), which vanishes if and only if  $\overline{\tau}_1 = \cdots = \overline{\tau}_W$.
\end{theorem}

\begin{remark} \label{remark:new_estimator}
Recall that the between-whole-plot additivity condition (\ref{eq:WP-additivity}) is equivalent to $\overline{\tau}_1 = \cdots =\overline{\tau}_W$ for every treatment contrast. Thus, even when the whole-plot sizes $M_1, \ldots, M_W$ are not all equal, by Theorem \ref{thm:zerobias}, the bias $\widetilde{\Delta}$ vanishes for every treatment contrast if and only if between-whole-plot additivity holds. Thus, if a positive semidefinite matrix $B$ satisfying conditions (c1)-(c3) is available, then Theorem \ref{thm:zerobias} provides us with a variance estimator that possesses properties similar to the one derived by \citet{zhaosplitplot::2018} for the balanced case. However, the issue of existence of such a matrix turns out to be quite challenging, and will be explored in the next section.
\end{remark}

\section{Existence and construction} \label{sec:existence_construction}

We will now study the existence of a positive semidefinite matrix $B$ satisfying conditions (c1)-(c3) stated in Theorem \ref{thm:zerobias} as a purely mathematical problem. 
Without loss of generality, we assume hereafter that
\begin{equation}
M_1 \le M_2 \le \cdots \le M_W. \label{eq:increasingM} 
\end{equation}
To motivate the ideas, consider first the case $W=3$, where conditions (c1) and (c2) determine $B$ uniquely as
\begin{equation}
B = \left [ \begin{array}{ccc}
M_1^2 & (M_3^2 - M_1^2 - M_2^2)/2 & (M_2^2 - M_1^2 - M_3^2)/2  \\
(M_3^2 - M_2^2 - M_1^2)/2 & M_2^2 & (M_1^2 - M_2^2 - M_3^2)/2  \\
(M_2^2 - M_3^2 - M_1^2)/2 & (M_1^2 - M_3^2 - M_2^2)/2 & M_3^2 
\end{array} \right]. \label{eq:Bfor3}
\end{equation}
This matrix is also positive semidefinite and satisfies (c3) if and only if its principal minor, given by the first two rows and columns, is positive. Simplification of this condition and application of (\ref{eq:increasingM}) yields $M_3 < M_1 + M_2$ as the necessary and sufficient condition for $B$ to satisfy (c1)-(c3). This construction of $B$ for $W=3$ raises the following questions with respect to the general case $W \ge 3$:
\begin{itemize}
\item[(a)] Is the condition
\begin{equation}
M_W < M_1 + \cdots + M_{W-1}, \label{eq:MwN-S}
\end{equation}
necessary and sufficient for existence of a positive semidefinite matrix $B$ satisfying (c1)-(c3)?
\item[(b)] If so, then under (\ref{eq:MwN-S}), can one construct such a matrix $B$ by an extension of the form in ( \ref{eq:Bfor3}) to the general case?
\end{itemize}

Later in this section, Theorem \ref{thm:Bexistence} answers (a) in the affirmative. On the other hand, the question in (b) does not allow a conclusive answer. To see why, observe that the most obvious extension of (\ref{eq:Bfor3}) to general $W \ge 3$ is given by $B = (( b_{w w^*} ))$, with
\begin{eqnarray}
&& b_{ww} = M_w^2, w = 1, \ldots, W, \nonumber \\
&& b_{w w^*} = \frac{M_1^2 + \cdots + M_W^2}{(W-1)(W-2)} - \frac{M_w^2 + M_{w^*}^2}{W-2}, \ w, w^* = 1, \ldots, W, \ w \ne w^*. \label{eq:B_gen}
\end{eqnarray} 
The divisors in (\ref{eq:B_gen}) ensure condition (c2) about zero row sums and make it consistent with (\ref{eq:Bfor3}) when $W = 3$. The form (\ref{eq:B_gen}) is also natural because, in keeping with $M_1^2, \ldots, M_W^2$ as the diagonal elements of $B$, it takes the off-diagonal elements as linear combinations of $M_1^2, \ldots, M_W^2$ in a systematic manner. However, unlike the case of $W = 3$, the matrix $B$ given by (\ref{eq:B_gen}) may not be positive semidefinite for $W \ge 4$, even when the condition (\ref{eq:MwN-S}) holds. For instance, if $W = 4$, then this condition holds for both the configurations $(M_1, \ldots, M_4) = (8,8,12,12)$ and $(6,6,14,14)$. The matrix $B$ in (\ref{eq:B_gen}) is positive semidefinite of  rank 3 ($= W-1$) for the first configuration, but has a negative eigenvalue for the second.

The above discussion makes it clear that, in general, the task of obtaining a positive semidefinite matrix $B$ satisfying (c1)-(c3) under condition (\ref{eq:MwN-S}) can be far more complex than what the form (\ref{eq:Bfor3}) arising for $W = 3$ suggests.Theorem \ref{thm:Bexistence} establishes condition (\ref{eq:MwN-S}) as a necessary and sufficient condition for existence of such a matrix. 

\begin{theorem} \label{thm:Bexistence}
Let $W \ge 3$. Then condition  (\ref{eq:MwN-S}), that is, $M_W < M_1 + \ldots + M_{W-1}$, is necessary and sufficient for the existence of a positive semidefinite matrix $B = ((b_{w w^*}))$ of order $W$ and satisfying the conditions
(c1) $b_{ww} = M_w^2, \ w=1, \ldots, W$, (c2) $\sum_{w^*=1}^W b_{w w^*} = 0, \ w=1, \ldots, W$, and (c3) rank$(B) = W-1$. 
\end{theorem}

The sufficiency part of the proof of Theorem \ref{thm:Bexistence} leads to a construction procedure of the matrix $B$ satisfying conditions (c1)-c(3). If $M_1 = \ldots = M_W (=M, \ \mbox{say})$, then one can simply take $M^2$ at each diagonal position of $B$ and $-M^2/(W-1)$ at each off-diagonal position. Turning next to the case of unequal $M_1 \le \ldots \le M_W$, suppose condition (\ref{eq:MwN-S}) holds. Let ${\mu} = (M_1, \ldots, M_{W-1})^{\prime}$, where the prime denotes transposition, and let $e$ denote the $(W-1) \times 1$ vector of ones. Then the steps involved in the construction of the matrix $B$ are:
\begin{itemize}
\item[]Step 1: Find a vector $x$ with elements $\pm 1$ satisfying the condition
\begin{equation}
|{\mu}^{\prime}x| < M_W. \label{eq:lem_x}
\end{equation}
\item[]Step 2: Find nonnegative constants $a_1$ and $a_2$, satisfying $a_1 + a_2 < 1$ and the following condition: 
\begin{equation}
a_1 \left\{ \left({\mu}^{\prime}x \right)^2 - {\mu}^{\prime}{\mu} \right\} + a_2 \left\{ \left({\mu}^{\prime}e \right)^2 - {\mu}^{\prime}{\mu} \right\} = M_W^2 - {\mu}^{\prime}{\mu}. \label{eq:lem_a}
\end{equation}
\item[]Step 3: Construct the following matrix:  
\begin{eqnarray*}
A = D \left\{ a_1 xx^{\prime} + a_2 ee^{\prime} + (1- a_1 - a_2) I \right \} D, \label{eq:matrix_A}
\end{eqnarray*}
where $x$, $a_1$ and $a_2$ are obtained from steps 1 and 2 above, $I$ is the identity matrix of order $W-1$ and $D = \text{diag}(M_1, \ldots, M_{W-1})$.
\item[]Step 4: Construct matrix $B$ as follows:
\begin{eqnarray*}
B = \left[ \begin{array}{cc} A & - A e \\
 - e^{\prime} A & e^{\prime} A e
\end{array}
\right] \label{eq:B_construct}
\end{eqnarray*}
\end{itemize}
Then $B$ is positive semidefinite of order $W$ and satisfies (c1)-(c3) by the proof of the sufficiency part of Theorem \ref{thm:Bexistence}. A lemma, crucial in this proof, appears in the supplementary material and guarantees the existence of vector $x$ in step 1 and constants $a_1$ and $a_2$ in step 2 under condition (\ref{eq:MwN-S}).

\begin{remark} \label{rem: generality}
It is satisfying that the condition (\ref{eq:MwN-S}) holds under wide generality. It only requires the largest whole-plot to be not too large compared to the others and holds, in particular, when there is a tie about the largest whole-plot. 
\end{remark}


\begin{remark} \label{rem:nonunique}
For $W=3$, one can check that the construction stated above yields the unique $B$ in (\ref{eq:Bfor3}). For $W \ge 4$, however, a positive semidefinite matrix $B$ meeting (c1)-(c3) is non-unique. Indeed, then the above construction itself can yield a wide class of such matrices $B$ considering all vectors $x$ which satisfy (\ref{eq:lem_x}), and for each such $x$, all nonnegative $a_1, a_2$ satisfying $a_1 + a_2 < 1$ and (\ref{eq:lem_a}). Thus, the issue of discriminating among rival choices of $B$ becomes important. Such a discriminating strategy is discussed in Section \ref{sec:minimax}.
\end{remark}


\section{Minimax estimators unbiased under between-whole-plot additivity} \label{sec:minimax}

As seen in Section \ref{sec:existence_construction}, while condition (\ref{eq:MwN-S}) guarantees the existence of matrix $B$ and consequently a variance estimator that is unbiased under between-whole-plot additivity, such a matrix is non-unique. Thus, it is important to define a criterion that can discriminate among possible choices of $B$. Clearly, a good choice should control the bias $\tilde{\Delta} = (1/N^2) \sum_{w=1}^W \sum_{w^*=1}^W b_{w w^*}\overline{\tau}_w \overline{\tau}_{w^*}$ given by (\ref{eq:newdelta}) that is associated with the estimation of $\var(\widehat{\overline{\tau}})$. The hurdle here is that, $\overline{\tau}_1, \ldots, \overline{\tau}_W$ are unknown. Even the idea of minimaxity does not work without further refinement, because $B$ is positive semidefinite, and hence $\tilde{\Delta}$ is unbounded with respect to variation of $\overline{\tau}_1, \ldots, \overline{\tau}_W$ in the $W$-dimensional real space. On the other hand, by (\ref{eq:connect_tau}), multiplication of $\overline{\tau}_1, \ldots, \overline{\tau}_W$ by any nonzero constant only rescales the treatment contrast $\overline{\tau}$, without essentially altering it. We, therefore, consider minimization of $\tilde{\Delta}$ subject to $\sum_{w=1}^W \overline{\tau}_w^2 = 1$. This is motivated by \cite{MDR::2018} who touched upon split-plot designs only in the balanced case. It is easy to see that the above formulation calls for obtaining $B$, subject to (c1)-(c3), so as to minimize $\lambda_{\max}(B)$, the largest eigenvalue of $B$. The following proposition provides us with a lower bound for $\lambda_{\max}(B)$.

\begin{proposition} \label{prop:minimax}
For any positive semidefinite matrix $B$ satisfying (c1)-(c3), a lower bound for $\lambda_{\max}(B)$ is given by $\lambda_0 = \sum_{w=1}^W M_w^2 /(W-1)$, but this bound is unattainable whenever $M_1, \ldots, M_W$ are not all equal.
\end{proposition}

Given Proposition \ref{prop:minimax}, an analytical solution to the minimaxity problem above seems to be intractable in the unbalanced case. This is anticipated, because a complete characterization of matrices $B$ satisfying (c1)-(c3) is hard, even though in Section \ref{sec:existence_construction}, we were able to outline a general method for constructing such matrices when condition (\ref{eq:MwN-S}) holds. As a practical strategy, therefore, it makes sense to concentrate on matrices $B$ that can be obtained via this method, with a view to minimizing $\lambda_{\max}(B)$ among these matrices. It is reassuring that even then the class of competing matrices $B$ is quite large, as noted in Remark \ref{rem:nonunique}.

\begin{example} \label{ex:Bschool}
Returning to the school example in Sections \ref{sec:intro} and \ref{sec:contrast_estimation}, where we have $N = 40$, $W=4$ and $(M_1, M_2, M_3, M_4) = (8,8,12,12)$, the smallest $\lambda_{\max}(B)$ obtainable via steps 1 through 4 described in Section \ref{sec:existence_construction} is 192, which corresponds to
\begin{eqnarray*}
B = \left [ \begin{array}{rrrr}
64 & 32 & -48 & -48 \\
32 & 64 & -48 & -48 \\
-48 & -48 & 144 & -48 \\
-48 & -48 & -48 & 144
\end{array} \right],
\end{eqnarray*}
as given by $x = (1, 1, -1)^{\prime}$, $a_1 = 0.5$ and  $a_2 = 0$.
\end{example}

\section{Simulation Results} \label{sec:simulations}

Whereas Theorem \ref{thm:zerobias} establishes unbiasedness of $\widetilde{V}(\widehat{\overline{\tau}})$ under (\ref{eq:MwN-S}) and between-whole-plot additivity, and consideration of minimaxity is expected to provide protection under extreme departures from additivity,  it is also important to understand how the bias of $\widetilde{V}(\widehat{\overline{\tau}})$ would compare to that of $\widehat{V}(\widehat{\overline{\tau}})$ under different levels of treatment effect heterogeneity. We now conduct some simulations to study this aspect. We consider the estimation of the interaction effect between factors $F_1$ and $F_2$ in the setting of Example \ref{ex:Bschool}. The unit-level treatment contrast $\tau_i$ equals $\{Y_i(00) - Y_i(01) - Y_i(10) + Y_i(11)\}/4$ for $i = 1, \ldots, 40$ \citep{DPR::2015}. The finite population contrast of interest is $\overline{\tau} = \sum_{i=1}^{40} \tau_i / 40$. The vector of potential outcomes for unit $i$, denoted by $Y_i = \left( Y_i(00),Y_i(01),Y_i(10),Y_i(11) \right)$, is generated using the multivariate normal model:
\begin{eqnarray*}
Y_i \sim N_4 \left( \theta_w, \Sigma_w \right), \ i \in \Omega_w, \ w=1, \ldots, 4,
\end{eqnarray*}
where
$$ \Sigma_w = \sigma_w^2 \left \{ (1-\rho_w) I_4 + \rho_w J_4 \right \} $$
is the covariance matrix for whole-plot $\Omega_w$ that depends on two parameters: the variance $\sigma_w^2$ and correlation $\rho_w$. Matrices $I_n$ and $J_n$ respectively denote the $n$th order identity matrix and the matrix of ones. Eight possible scenarios (listed in Table \ref{tab:settings}) for generating the potential outcomes are considered.

\begin{table}[ht] 
\centering \footnotesize
\caption{Simulation settings} \label{tab:settings}
\begin{tabular}{c|cccc|cccc|rrrr} \hline
Population & $\theta_1$ & $\theta_2$ & $\theta_3$ & $\theta_4$ & $\sigma^2_1$ & $\sigma^2_2$ & $\sigma^2_3$ & $\sigma^2_4$ & $\rho_1$ & $\rho_2$ & $\rho_3$ & $\rho_4$ \\ \hline
I  & (10,5,9,8) & (10,5,9,8) & (10,5,9,8) & (10,5,9,8) & 2 & 2 & 2 & 2 & 1 & 1 & 1 & 1 \\
II  & (10,5,9,8) & (9,7,4,6) & (11,8,7,8) & (8,7,6,9) & 2.5 & 2 & 2 & 3 & .5 & .5 & .5 & .5 \\
III & (10,5,9,8) & (5,9,10,8) & (10,9,8,5) & (10,5,8,9) & 2.5 & 2 & 2 & 3 & 1 & 1 & 1 & 1 \\
IV & (10,5,9,8) & (5,9,10,8) & (10,9,8,5) & (10,5,8,9) & 2.5 & 2 & 2 & 3 & .5 & .5 & .5 & .5 \\
V & (10,5,9,8) & (5,9,10,8) & (10,9,8,5) & (10,5,8,9) & 2.5 & 2 & 2 & 3 & .2 & .4 & .6 & .8 \\
VI & (10,5,9,8) & (5,9,10,8) & (10,9,8,5) & (10,5,8,9) & 2.5 & 2 & 2 & 3 & 0 & 0 & 0 & 0 \\
VII & (10,5,9,8) & (5,9,10,8) & (10,9,8,5) & (10,5,8,9) & 2.5 & 2 & 2 & 3 & -.3 & -.3 & -.3 & -.3  \\
VIII & (10,5,9,8) & (5,9,10,8) & (10,9,8,5) & (10,5,8,9) & 2.5 & 2 & 2 & 3 & -.3 & .3 & -.3 & .3 \\ \hline
\end{tabular}
\end{table}

\begin{figure}[h]
\centering
\begin{tabular}{c}
\includegraphics[width=6 in]{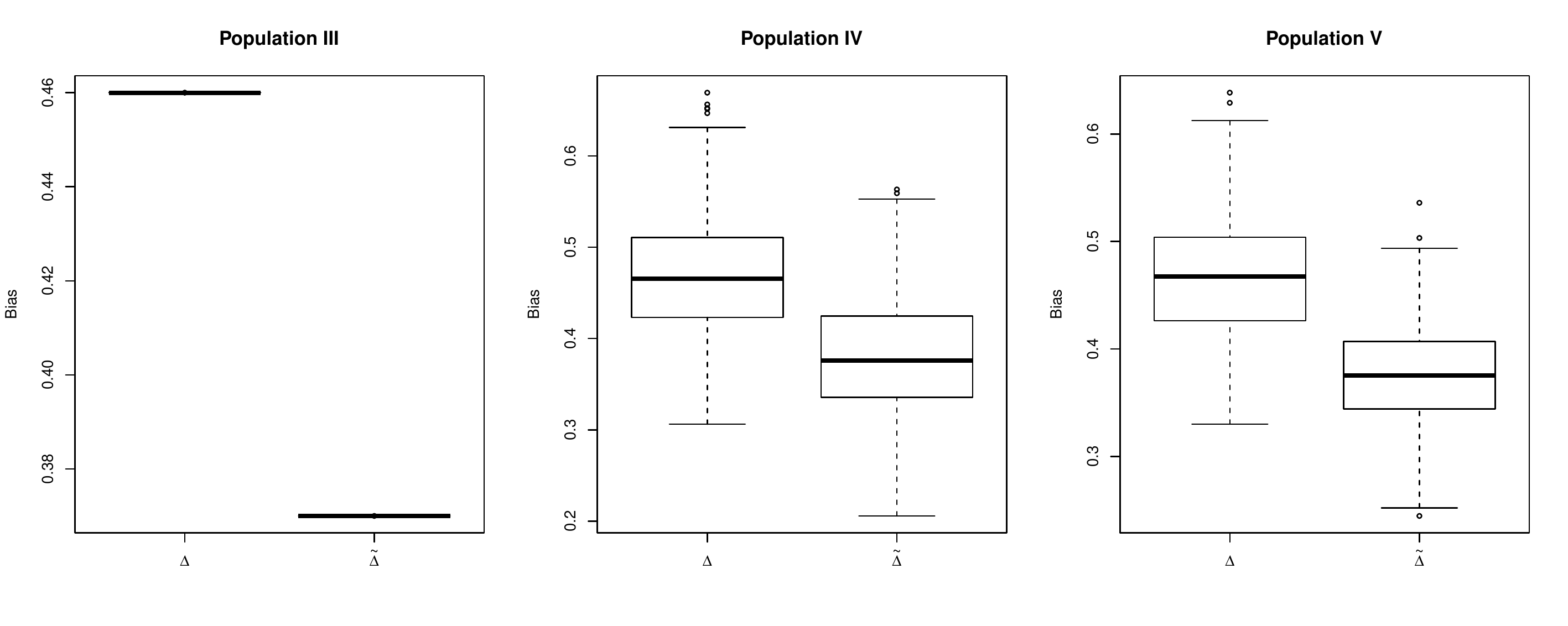} \\
\includegraphics[width=6 in]{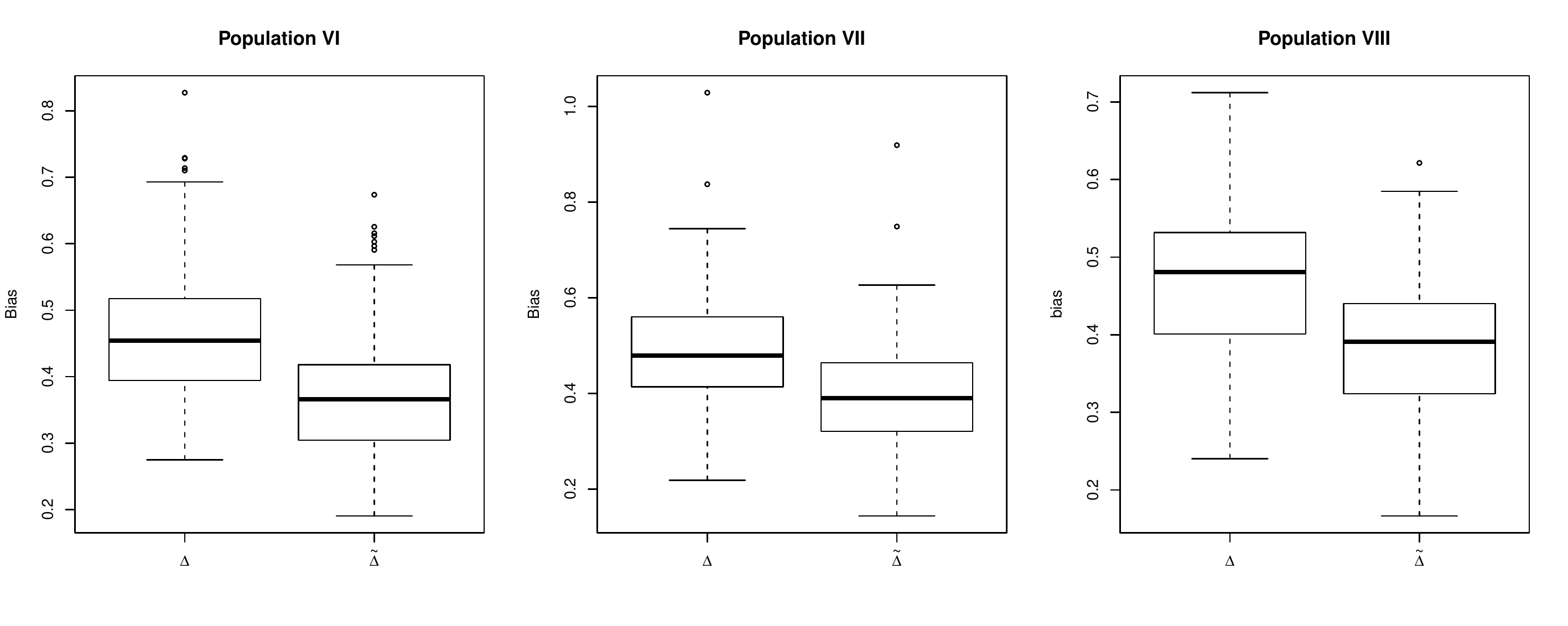}
\end{tabular}
\caption{Boxplots of $\Delta$ and $\widetilde{\Delta}$ for populations III-VIII}
\label{fig:boxplots}
\end{figure}

Strict additivity holds for population I. The potential outcomes for population II are forced to to ensure, via an appropriate command in R, that the whole-plot means $\overline{\tau}_1, \ldots, \overline{\tau}_4$ are always one. Population III generates different $\overline{\tau}_1, \ldots, \overline{\tau}_4$ but guarantees the same $\tau_i$ within each whole-plot. Populations IV through VIII differ only with respect to the correlation parameters that lead to different types of treatment effect heterogeneity. These include all zero correlations in population VI, all negative correlations in population VII, and a mix of positive and negative correlations in population VIII.

From each population, 200 sets of potential outcomes are generated, and the biases of variance estimators $\widehat{V}(\widehat{\overline{\tau}})$ and $\widetilde{V}(\widehat{\overline{\tau}})$ are compared. Note that these biases are $\Delta$ given by (\ref{eq:delta}) and $\tilde{\Delta}$ given by (\ref{eq:newdelta}). We also calculate the bias ratio $\widetilde{\Delta}/\Delta$ for each population. The results for populations I and II are consistent with our results. In both of these cases, $\widetilde{\Delta}$ is always zero and $\Delta$ is always 0.0133. Boxplots of the distributions of $\Delta$ and $\tilde{\Delta}$ for populations III-VIII are shown in Figure \ref{fig:boxplots}. The median bias ratios for these populations are 0.804, 0.811, 0.811, 0.810, 0.822 and 0.817 respectively. The plots and the median bias ratios establish the robustness of the new estimator $\widetilde{V}(\widehat{\overline{\tau}})$ with respect to controlling bias under various forms of treatment effect heterogeneity.



\section*{Acknowledgement}

This work was supported by the J.C. Bose National Fellowship, Government of India, and grants from Indian Institute of Management Calcutta and National Science Foundation, USA.

\section*{Appendix: Proofs of results} \label{Appendix}

In what follows, $E_1$ and $\cov_1$ denote unconditional expectation and covariance with respect to the randomization at the whole-plot stage, while $E_2$ and $\cov_2$ denote expectation and covariance with respect to the randomization at the sub-plot stage, conditional on the whole-plot stage assignment. 

\begin{proof}[Proof of Proposition 1] 

Follows from straightforward conditioning arguments.

\end{proof}

\begin{proof}[Proof of Theorem 1]

Recall that
\begin{eqnarray*}
\overline{U}_w^{\obs}(z_1 z_2) &=& \frac{1}{r_{w2}(z_2)} \sum_{i \in T_{w2}(z_2)} U_i(z_1 z_2), w \in T_1(z_1) \\ 
\overline{U}^{\obs}(z_1 z_2) &=& \frac{1}{r_1(z_1)} \sum_{w \in T_1(z_1)}\overline{U}_w^{\obs}(z_1 z_2).
\end{eqnarray*}
Consequently,
\begin{eqnarray*}
E_2 \left\{ \overline{U}^\obs(z_1 z_2) \right\} &=& \frac{1}{r_1(z_1)} \sum_{w \in T_1(z_1)} \overline{U}_w (z_1 z_2), \quad \mbox{and} \quad \\
E_2 \left\{ \overline{U}^\obs(z_1^* z_2^*) \right\} &=& \frac{1}{r_1(z_1^*)} \sum_{w \in T_1(z_1^*)} \overline{U}_w (z_1^* z_2^*).
\end{eqnarray*}
Defining $\delta(z_1, z_1^*)$ as an indicator that equals one if $z_1 = z_1^*$ and zero otherwise, we have
\begin{eqnarray*}
&& \cov_1 \left[ E_2 \left\{ \overline{U}^\obs(z_1 z_2) \right\}, E_2 \left\{ \overline{U}^\obs(z_1^* z_2^*) \right\} \right] \\
&=& \frac{1}{(W-1) W r_1(z_1)} \sum_{w=1}^W \left \{ \overline{U}_w (z_1 z_2) - \overline{U} (z_1 z_2) \right \} \left \{ \overline{U}_w (z_1^* z_2^*) - \overline{U} (z_1^* z_2^*) \right \} 
\left \{ W \delta(z_1, z_1^*) - r_1(z_1) \right \}, \\
&=& \frac{1}{W \overline{M} r_1(z_1)} S_{\bt}(z_1 z_2, z_1^* z_2^*) \left \{ W \delta(z_1, z_1^*) - r_1(z_1) \right \}. 
\end{eqnarray*}

\noindent Next,
\begin{eqnarray*}
&& \cov_2 \left\{ \overline{U}^\obs(z_1 z_2), \overline{U}^\obs(z_1^* z_2^*) \right\} \\
&=& \delta(z_1, z_1^*) \sum_{w \in T_1(z_1)} \frac{  S_{\win,w}(z_1 z_2, z_1^* z_2^*) \left \{ M_w \delta(z_2, z_2^*) - r_{w2}(z_2) \right \}} {M_w r_{w2}(z_2) \left\{ r_1(z_1) \right\}^2 }. 
\end{eqnarray*}
so that
\begin{eqnarray*}
&& E_1 \left[ \cov_2 \left\{ \overline{U}^\obs(z_1 z_2), \overline{U}^\obs(z_1^* z_2^*) \right\} \right] \\
&=& \delta(z_1, z_1^*) \sum_{w=1}^W \frac{  S_{\win,w}(z_1 z_2, z_1^* z_2^*) \left \{ M_w \delta(z_2, z_2^*) - r_{w2}(z_2) \right \}} {W M_w r_1(z_1) r_{w2}(z_2) }. 
\end{eqnarray*}
Hence,
\begin{eqnarray}
&& \cov \left\{ \overline{U}^\obs(z_1 z_2), \overline{U}^\obs(z_1^* z_2^*) \right\}  \nonumber \\
&=& \delta(z_1, z_1^*) \left\{  \frac{S_{\bt}(z_1 z_2, z_1^* z_2^*)}{\overline{M} r_1(z_1)} - \sum_{w=1}^W \frac{  S_{\win,w}(z_1 z_2, z_1^* z_2^*)}{W M_w r_1(z_1)}  \right\} \nonumber \\ 
&+& \delta(z_1, z_1^*) \delta(z_2, z_2^*) \sum_{w=1}^W \frac{S_{\win,w}(z_1 z_2, z_1^* z_2^*)}{W r_1(z_1) r_{w2}(z_2)} - \frac{S_{\bt}(z_1 z_2, z_1^* z_2^*)}{N}. \label{eq:cov}
\end{eqnarray}

Since $\widehat{\overline{\tau}} = \sum_{z_1 \in Z_1} \sum_{z_2 \in Z_2} g(z_1 z_2) \overline{U}^\obs(z_1 z_2)$, we have that
\begin{eqnarray*}
\var( \widehat{\overline{\tau}}) = \sum_{z_1 \in Z_1} \sum_{z_2 \in Z_2} \sum_{z_1^* \in Z_1} \sum_{z_2^* \in Z_2}  g(z_1 z_2)  g(z_1^* z_2^*)\cov \left\{ \overline{U}^\obs(z_1 z_2), \overline{U}^\obs(z_1^* z_2^*) \right\}.
\end{eqnarray*}

Substituting the expression of $\cov \left\{ \overline{U}^\obs(z_1 z_2), \overline{U}^\obs(z_1^* z_2^*) \right\}$ from (\ref{eq:cov}) in the above, the first two terms in the expression of $\var( \widehat{\overline{\tau}})$ in Theorem 1 follow immediately. The last term can be explained as
\begin{eqnarray*}
&& \sum_{z_1 \in Z_1} \sum_{z_2 \in Z_2} \sum_{z_1^* \in Z_1} \sum_{z_2^* \in Z_2} g(z_1 z_2) g(z_1^* z_2^*) S_{\bt}(z_1 z_2, z_1^* z_2^*) / N \\
&=& \frac{ \overline{M} }{(W-1)N}  \sum_{w=1}^W \left[ \sum_{z_1 \in Z_1} \sum_{z_2 \in Z_2} g(z_1 z_2) \left\{\overline{U}_w (z_1 z_2) - \overline{U} (z_1 z_2)  \right\} \right]^2 \\
&=& \frac{1}{W(W-1)} \sum_{w=1}^W \left[ \sum_{z_1 \in Z_1} \sum_{z_2 \in Z_2} g(z_1 z_2) \left\{ (M_w/\overline{M})\overline{Y}_w (z_1 z_2) - \overline{Y} (z_1 z_2)  \right\} \right]^2 \\
&=& \frac{1}{W(W-1)} \sum_{w=1}^W \left\{ (M_w/\overline{M})\overline{\tau}_w - \overline{\tau} \right \}^2 = \Delta.
\end{eqnarray*}

\end{proof}

\begin{proof}[Proof of Theorem 2]

\begin{eqnarray*}
E_2 \left\{ \widehat{S}(z_1 z_2, z_1 z_2^*) \right\} &=& \frac{1}{r_1(z_1) }\sum_{w \in T_1(z_1)} \cov_2 \left\{ \overline{U}_w^\obs(z_1 z_2), \overline{U}_w^\obs(z_1 z_2^*) \right\} \\
&+& \frac{1}{r_1(z_1)-1} \sum_{w \in T_1(z_1)} \left \{ \overline{U}_w(z_1 z_2) - \widetilde{\overline{U}}(z_1 z_2) \right\} \left \{ \overline{U}_w(z_1 z_2^*) - \widetilde{\overline{U}}(z_1 z_2^*) \right\},
\end{eqnarray*}
where $\widetilde{\overline{U}}(z_1 z_2) = \sum_{w \in T_1(z_1)} \overline{U}_w(z_1 z_2) / r_1(z_1)$, and $\widetilde{\overline{U}}(z_1 z_2^*)$ is similarly defined. For any $w \in T_1(z_1)$,
\begin{eqnarray*}
\cov_2 \left\{ \overline{U}_w^\obs(z_1 z_2), \overline{U}_w^\obs(z_1 z_2^*) \right\} = \frac{ S_{\win,w}(z_1 z_2, z_1 z_2^*) \left\{ M_w \delta(z_2, z_2^*) - r_{w2}(z_2) \right\}}{M_w r_{w2}(z_2)}.
\end{eqnarray*}
Thus,
\begin{eqnarray*}
E \left\{ \widehat{S}(z_1 z_2, z_1 z_2^*) \right\} &=& E_1 E_2 \left\{ \widehat{S}(z_1 z_2, z_1 z_2^*) \right\} \\
&=& \sum_{w=1}^W \frac{ S_{\win,w}(z_1 z_2, z_1 z_2^*) \left\{ M_w \delta(z_2, z_2^*) - r_{w2}(z_2) \right\}}{W M_w r_{w2}(z_2)} \\
&+& \frac{1}{W-1} \sum_{w=1}^W \left\{ \overline{U}_w(z_1 z_2) - \overline{U}(z_1 z_2)  \right\} \left\{ \overline{U}_w(z_1 z_2^*) - \overline{U}(z_1 z_2^*)  \right\} \\
&=& \sum_{w=1}^W \frac{ S_{\win,w}(z_1 z_2, z_1 z_2^*) \left\{ M_w \delta(z_2, z_2^*) - r_{w2}(z_2) \right\}}{W M_w r_{w2}(z_2)} + \frac{S_{\bt}(z_1 z_2, z_1 z_2^*)}{\overline{M}}. 
\end{eqnarray*}
The result stated in Theorem 2 is evident from the above.
 \end{proof}

\begin{proof}[Proof of Proposition 2]
Because $w \ne w^*$, by (5) and the definition of $G_w^{\text{obs}}$, conditionally on the assignment of the whole-plots to the level combinations of the whole-plot factors, $G_w^{\obs}$ and $G_{w^*}^{\obs}$ are independent and the conditional expectation of their product equals
$$ \left\{ \sum_{z_2 \in Z_2} g(z_{1w} z_2) \overline{Y}_w (z_{1w} z_2)  \right\} \left\{ \sum_{z_2 \in Z_2} g(z_{1w^*} z_2) \overline{Y}_{w^*} (z_{1w^*} z_2)  \right\}. $$  
The result now follows from (9), noting that the pair $(z_{1w}, z_{1w^*})$ equals any $(z_1, z_1^*)$ with probability
$ \frac{r_1(z_1) \left\{ r_1(z_1^*) - \delta(z_1, z_1^*) \right\}}{W(W-1)}$.
\end{proof}

\bigskip

\begin{proof}[Proof of the necessity part of Theorem 4]
Suppose a positive semidefinite matrix $B = (b_{w w^*})$ of order $W$ and satisfying (c1)-(c3) exists. Then by (c1),
\begin{equation}
|b_{w w^*}| \le M_w M_{w^*}, \quad w, w^* = 1, \ldots, W, \quad w \ne w^*. \label{eq:bound}
\end{equation}
Hence using (c2), (\ref{eq:bound}), and (c1) in succession,
\begin{equation}
0 = b_{W1} + \ldots + b_{WW} \ge b_{WW} - M_W (M_1 + \ldots + M_{W-1}) = M_W (M_W - M_1 - \ldots - M_{W-1}), \label{eq:thm4_2}
\end{equation}
which implies $M_W \le M_1 + \ldots + M_{W-1}$. If possible, let equality hold here. Then equality holds throughout in (\ref{eq:thm4_2}), and invoking (\ref{eq:bound}), this yields
\begin{equation}
b_{Ww} = - M_W M_w, \quad w=1, \ldots, W-1. \label{eq:eq:thm4_3}
\end{equation}
For any $w, w^*$ such that $w < w^* < W$, by (c1) and (\ref{eq:eq:thm4_3}), the principal minor of $B$, as given by its $w$th, $w^*$th and $W$th rows and columns turns out to be $-M_W^2(b_{w w^*} - M_w M_{w^*})^2$. Because this principal minor is nonnegative due to positive semidefinite-ness of $B$, it follows that $b_{w w^*} = M_w M_{w^*}$. This, in conjunction with (c1) and (\ref{eq:eq:thm4_3}), implies that $B = b b^{\prime}$, where $b = (M_1, \ldots, M_{W-1}, -M_W)^{\prime}$. But then $\rank(B) = 1 < W-1$, and (c3) is violated. This contradiction proves the necessity of the condition $M_W < M_1 + \ldots + M_{W-1}$. 
\end{proof}

\bigskip

To prove the sufficiency part of Theorem 4, we first state a lemma that is crucial in this proof and also leads to the algorithm for construction of the symmetric positive semidefinite matrix $B$ of order $W$ that satisfies conditions (c1)-(c3).

\begin{lemma}\label{lemma:sufficiency}
Let $W \ge 3$. Suppose $M_1, \ldots, M_W$ are not all equal and $M_1 \le \ldots \le M_W$, as per (19). Let $e$ denote the $(W-1) \times 1$ vector of ones and $\mu = (M_1, \ldots, M_{W-1})^{\prime}$.
\begin{itemize}
\item[(a)] Then there exists a $(W-1) \times 1$ vector $x$ with elements $\pm 1$ such that $|\mu^{\prime}x| < M_W$.
\item[(b)] If, in addition, condition (21) holds, i.e., $M_W < M_1 + \ldots + M_{W-1}$, then, with the vector $x$ as in (a) above, there exist nonnegative constants $a_1, a_2$ satisfying $a_1 + a_2 < 1$, such that equation (24) holds, i.e.,
\begin{eqnarray*}
a_1 \left\{ \left(\mu^{\prime}x \right)^2 - \mu^{\prime}\mu \right\} + a_2 \left\{ \left(\mu^{\prime}e \right)^2 - \mu^{\prime}\mu \right\} = M_W^2 - \mu^{\prime}\mu. \label{eq:lem_a}
\end{eqnarray*}
\end{itemize}
\end{lemma}

\medskip

\begin{proof} [Proof of Lemma \ref{lemma:sufficiency}]

\noindent Part (a). It will suffice to show that there exist $x_1, \ldots, x_{W-1}$, each $+1$ or $-1$, such that $|\sum_{w=1}^{W-1} M_w x_w| < M_W$. One can then simply take $x = (x_1, \ldots, x_{W-1})^{\prime}$. 
Recall that $M_1 \le M_2 \le \ldots \le M_W$, as per (19). Because $M_1 \ldots,M_W$ are not all equal, this yields 
\begin{equation}
M_1 < M_W.			\label{eq:lem1}
\end{equation}
Let $h$ be the largest nonnegative integer such that
\begin{equation} 
 M_{W-2h} = M_W. \label{eq:lem2}
\end{equation}
By (\ref{eq:lem1}), $W-2h \ge 2$. If $h \ge 1$, define
\begin{equation}
x_{W-h} = \ldots = x_{W-1} = 1, \quad x_{W-2h} = \ldots = x_{W-h-1} = -1, \label{eq:lem3}
\end{equation}
and note that
\begin{equation}
\sum_{w=W-2h}^{W-1} M_w x_w = 0, \label{eq:lem4}
\end{equation}
because by (19) and (\ref{eq:lem2}), $M_w = M_W$ for $w=W-2h, \ldots, W-1$. Now, if $W-2h = 2$, then with $x_1 = 1$ and $x_2, \ldots, x_{W-1}$ as in (\ref{eq:lem3}), $|\sum_{w=1}^{W-1} M_w x_w| = M_1 < M_W$, by (\ref{eq:lem1}) and (\ref{eq:lem4}).

Next, let $W-2h \ge 3$. Then, by (19), 
$$ \sum_{w=2}^{W-2h-1} M_w \ge (W-2h-2)M_2 \ge M_1. $$
Let $w_1$ be the largest integer in $\{1, \ldots, W-2h-2\}$ such that $\sum_{w=1}^{w_1} M_w \le \sum_{w=w_1+1}^{W-2h-1} M_w$. If $w_1 = W-2h-2$, then $\sum_{w=1}^{W-2h-2} M_w \le M_{W-2h-1}$. So, with $x_1 = \ldots = x_{W-2h-2} = -1$, $x_{W-2h-1} = 1$ and $x_{W-2h}, \ldots, x_{W-1}$ as in (\ref{eq:lem3}) when $h \ge 1$,
$$ \left|\sum_{w=1}^{W-1} M_w x_w \right| = M_{W-2h-1} - \sum_{w=1}^{W-2h-2} M_w < M_{W-2h-1} \le M_W, $$ by (\ref{eq:lem4}).

Now, suppose $1 \le w_1 \le W-2h-3$, in which case $W-2h \ge 4$. Then,
$$ \sum_{w=1}^{w_1} M_w \le \sum_{w=w_1 +1}^{W-2h-1} M_w, \quad \mbox{and} \quad \sum_{w=1}^{w_1+1} M_w > \sum_{w=w_1 +2}^{W-2h-1} M_w. $$
As a result, either 
$$ \mbox{(i)} \quad \left| \sum_{w=w_1 +1}^{W-2h-1} M_w - \sum_{w=1}^{w_1} M_w \right| < M_W 
\quad \mbox{or (ii)} \quad \left| \sum_{w=1}^{w_1+1} M_w - \sum_{w=w_1+2}^{W-2h-1} M_w \right| < M_W. $$ 
Else, 
$$ \sum_{w=w_1 +1}^{W-2h-1} M_w - \sum_{w=1}^{w_1} M_w \ge M_W, \quad \mbox{as well as} \quad \sum_{w=1}^{w_1+1} M_w - \sum_{w=w_1+2}^{W-2h-1} M_w \ge M_W. $$
Adding these two inequalities, we have $M_{w_1+1} \ge M_W$, which is impossible by the definition of $h$, because $w_1+1 \le W-2h-2$.

If (i) holds, then the choice $x_1 = \ldots = x_{w_1} = -1$, $x_{w_1+1} = \ldots = x_{W-2h-1} = 1$, coupled with $x_{W-2h}, \ldots, x_{W-1}$ as in (\ref{eq:lem3}) when $h \ge 1$, entails $\left| \sum_{w=1}^{W-1} M_w x_w \right| < M_W$, by (\ref{eq:lem4}). Similarly, if (ii) holds, then the choice $x_1 = \ldots = x_{w_1+1} = -1$, $x_{w_1+2} = \ldots = x_{W-2h-1} = 1$, coupled with $x_{W-2h}, \ldots, x_{W-1}$ as in (\ref{eq:lem3}) when $h \ge 1$, entails $\left| \sum_{w=1}^{W-1} M_w x_w \right| < M_W$.

\bigskip

\noindent Part (b): Let $M_W < M_1 + \ldots + M_{W-1} = \mu^{\prime} e$, and let the vector $x$ be as in part (a) above, so that $| \mu^{\prime} x| < M_W$. Let $\phi_1 = \left( \mu^{\prime} x \right)^2 -  \mu^{\prime} \mu$, $\phi = M_W^2 - \mu^{\prime} \mu$ and $\phi_2 =  \left( \mu^{\prime} e \right)^2 - \mu^{\prime} \mu$. Then $\phi_1 < \phi < \phi_2$, as $|\mu^{\prime} x| < M_W < \mu^{\prime} e$. As a result, there exist constants $\tilde{a}_1$ and $\tilde{a}_2$ such that $0 \le \tilde{a}_1, \tilde{a}_2 < 1$ and $\tilde{a}_1 \phi_1 < \phi < \tilde{a}_2 \phi_2$. Let $\xi = \left( \tilde{a}_2 \phi_2 - \phi \right)/ \left( \tilde{a}_2 \phi_2 - \tilde{a}_1 \phi_1 \right)$. Then $0 < \xi < 1$. Hence, if we take $a_1 = \tilde{a}_1 \xi$, $a_2 = \tilde{a}_2 (1-\xi)$, then $a_1, a_2 \ge 0$ and $a_1 + a_2 < 1$, because $a_1 + a_2$ is a weighted average of $\tilde{a}_1$ and $\tilde{a}_2$, both of which are less than one. Moreover, $a_1 \phi_1 + a_2 \phi_2 = \phi$ by the definition of $\xi$, i.e., $a_1$ and $a_2$ satisfy (24).
\end{proof}

\bigskip

\begin{proof}[Proof of the sufficiency part of Theorem 4] 

In view of Lemma \ref{lemma:sufficiency}, this follows from steps 1-4 in Section 5, noting that (i) the matrix $A$ there is positive definite, and hence the matrix $B$ there is positive semidefinite of rank $W-1$ with each row sum zero, (ii) $A$ has diagonal elements $M_1^2, \ldots, M_{W-1}^2$, and (iii) by (24),
$$ e^{\prime} A e = a_1 (\mu^{\prime} x)^2 + a_2 (\mu^{\prime} e)^2 + (1-a_1 - a_2) \mu^{\prime} \mu = M_W^2,$$
because $De = \mu$.

\end{proof}

\bibliographystyle{apalike}
\bibliography{ACE}

\begin{thebibliography}{}

\bibitem[Box et~al., 2005]{BHH::2005}
Box, G. E.~P., Hunter, J.~S., and Hunter, W.~G. (2005).
\newblock {\em Statistics for Experimenters: Design, Innovation, and
  Discovery}.
\newblock John Wiley \& Sons, Hoboken, New Jersey, 2nd edition.

\bibitem[Cochran, 1977]{cochran::1977}
Cochran, W.~G. (1977).
\newblock {\em Sampling Techniques}.
\newblock John Wiley \& Sons: New York.

\bibitem[Cochran and Cox, 1957]{CC::1957}
Cochran, W.~G. and Cox, G.~M. (1957).
\newblock {\em Experimental Designs}.
\newblock John Wiley \& Sons, Hoboken, New Jersey, 2nd edition.

\bibitem[Dasgupta et~al., 2015]{DPR::2015}
Dasgupta, T., Pillai, N.~S., and Rubin, D.~B. (2015).
\newblock Causal inference for $2^{K}$ factorial designs by using potential
  outcomes.
\newblock {\em Journal of the Royal Statistical Society, Series B},
  77(4):727--753.

\bibitem[Fisher, 1925]{Fisher1925}
Fisher, R.~A. (1925).
\newblock {\em Statistical Methods for Research Workers}.
\newblock Oliver \& Boyd, Edinburgh, Scotland.

\bibitem[Fisher, 1935]{Fisher1935}
Fisher, R.~A. (1935).
\newblock {\em The Design of Experiments}.
\newblock Oliver \& Boyd, Oxford, England, 1st edition.

\bibitem[Freedman, 2006]{Freedman2006}
Freedman, D.~A. (2006).
\newblock Statistical models for causation: What inferential leverage do they
  provide?
\newblock {\em Evaluation Review}, 30:691--713.

\bibitem[Freedman, 2008]{Freedman2008a}
Freedman, D.~A. (2008).
\newblock On regression adjustments to experimental data.
\newblock {\em Advances in Applied Mathematics}, 40:180--193.

\bibitem[Kirk, 1982]{Kirk::1982}
Kirk, R.~E. (1982).
\newblock {\em Experimental Design: Procedures for the Behavioral Sciences}.
\newblock Brooks/Cole, Monterey, CA.

\bibitem[Mukerjee et~al., 2018]{MDR::2018}
Mukerjee, R., Dasgupta, T., and Rubin, D.~B. (2018).
\newblock Randomization-based causal inference from split-plot designs.
\newblock {\em Journal of the American Statistical Association}, 113:868--881.

\bibitem[Neyman, 1923]{neyman::1923}
Neyman, J. (1923).
\newblock On the application of probability theory to agricultural experiments.
  {E}ssay on principles. {S}ection 9.
\newblock {\em Statistical Science}, 5:465--472.

\bibitem[Rubin, 1974]{Rubin1974}
Rubin, D.~B. (1974).
\newblock Estimating causal effects of treatments in randomized and
  nonrandomized studies.
\newblock {\em Journal of Educational Psychology}, 66:688--701.

\bibitem[Rubin, 1978]{Rubin1978}
Rubin, D.~B. (1978).
\newblock Bayesian inference for causal effects: The role of randomization.
\newblock {\em The Annals of Statistics}, 6:34--58.

\bibitem[Rubin, 2005]{Rubin2005}
Rubin, D.~B. (2005).
\newblock Causal inference using potential outcomes: Design, modeling,
  decisions.
\newblock {\em Journal of the American Statistical Association}, 100:322--331.

\bibitem[Wu and Hamada, 2009]{WandH::2009}
Wu, C. F.~J. and Hamada, M.~S. (2009).
\newblock {\em Experiments: Planning, Analysis, and Optimization}.
\newblock John Wiley \& Sons, Hoboken, New Jersey, 2nd edition.

\bibitem[Yates, 1935]{Yates::1935}
Yates, F. (1935).
\newblock Complex experiments.
\newblock {\em Supplement to the Journal of the Royal Statistical Society},
  2:181--247.

\bibitem[Zhao et~al., 2018]{zhaosplitplot::2018}
Zhao, A., Ding, P., Mukerjee, R., and Dasgupta, T. (2018).
\newblock Randomization-based causal inference from split-plot designs.
\newblock {\em Annals of Statistics}, 46:1876--1903.

\end{thebibliography}

\end{document}